\documentclass[a4paper]{article}
\usepackage[utf8]{inputenc}
\usepackage{amsmath}
\usepackage{amssymb}
\usepackage{amsthm}
\usepackage{authblk} 
\usepackage{color}
\usepackage{graphicx}
\usepackage[colorinlistoftodos]{todonotes}
\usepackage[colorlinks=true, allcolors=blue]{hyperref}
\usepackage[numbers]{natbib}

\newtheorem{theorem}{Theorem}[section]
\newtheorem{lemma}[theorem]{Lemma}
\newtheorem{proposition}[theorem]{Proposition}
\newtheorem{corollary}[theorem]{Corollary}
\newtheorem{definition}[theorem]{Definition}

\newtheorem{question}[theorem]{Question}

\newcommand{\complex}{\mathbb{C}}

\begin{document}

\title{Maximally entangled correlation sets}

\date{}


\date{June 2018. Revised June 2019.}



\author[ ]{Elie Alhajjar}
\author[ ]{Travis B. Russell}

\affil[ ]{Army Cyber Institute, United States Military Academy, West Point, NY 10996}
\affil[ ]{ \textit{elie.alhajjar@westpoint.edu, travis.russell@westpoint.edu}}





\maketitle 





\begin{abstract} We study the set of quantum correlations generated by actions on maximally entangled states. We show that such correlations are dense in their own convex hull. As a consequence, we show that these correlations are dense in the set of synchronous quantum correlations. We introduce the concept of corners of correlation sets and show that every local or nonsignalling correlation can be realized as the corner of a synchronous local or nonsignalling correlation. We provide partial results for other correlation sets. \end{abstract}









\section{Introduction}

A correlation is  a tuple $(p(i,j|x,y))$ of positive real numbers such that for each choice of $x$ and $y$ one obtains a joint discrete probability distribution $(p(i,j|x,y))_{i,j}$. In this context, the value $p(i,j|x,y)$ represents the probability that two actors, usually named Alice and Bob, obtain outcomes $i$ and $j$ (respectively) given that they performed experiments $x$ and $y$ (respectively). We let $C(n_A,n_B,m_A,m_B)$ denote the set of all correlations, where Alice and Bob may perform $n_A$ or $n_B$ experiments (respectively) and each experiment has $m_A$ or $m_B$ possible outcomes (respectively). Usually we are concerned with subsets denoted by $C_r(n_A,n_B,m_A,m_B)$ arising from different probabilistic models denoted by the variable $r$. For simplicity, we write $C_r(n_A,n_B,m)$ when $m_A = m_B$ and $C_r(n,m)$ when $n_A=n_B$ and $m_A=m_B$, or simply $C_r$ when the numbers of experiments and outcomes are unspecified or clear from context. Of principal interest in this paper is the set of quantum correlations, denoted by $C_q$.

The study of quantum correlations goes back to foundational questions in physics posed by Einstein-Podolsky-Rosen \cite{EPR}. These questions are equivalent to asking whether or not the set of correlations arising from a theory of local hidden variables, denoted by $C_{loc}$, coincides with the set of quantum correlations, denoted by $C_q$. These questions were settled by John Bell who showed \cite{Bell_EPR} that $C_q \neq C_{loc}$. In subsequent decades, Tsirelson began asking similar questions concerning the relationship between $C_q$ and the set of correlations attainable in relativistic quantum theory \cite{Tsirel'son1987}, which we denote by $C_{qc}$. In short, Tsirelson's yet unanswered question (the weak Tsirelson conjecture) asks whether or not $C_q$ is dense in $C_{qc}$.

Over the past few years much progress has been made in understanding the geometry of the quantum correlation sets, although many open problems remain. For example, Slofstra proved  \cite{Slofstra1} that $C_q(184,235,8,2)$ is not closed and hence showed that the quantum correlation sets are not closed in general, a question that had been open for some time. By studying the structure of synchronous quantum correlations, Dykema-Paulsen-Prakash showed  \cite{DPP1} that the quantum correlation sets $C_q(n,m)$ are not closed for $n \geq 5$ and $m \geq 2$. Kim-Paulsen-Schafhauser showed  \cite{MR3776034} that the synchronous quantum correlations coincide with the synchronous quantum spacial correlations. They also gave a positive answer to the synchronous approximation problem of Dykema-Paulsen  \cite{MR3432742}, proving that Connes' embedding conjecture \cite{ConnesConjecture}, a much-studied problem open since the 1970s, is equivalent to showing that the synchronous quantum commuting correlations coincide with the closure of the synchronous quantum correlations. In spite of all these breakthroughs, the geometry of the quantum correlation sets is not fully understood in the literature.

In this paper, we provide a new approach to describing synchronous quantum correlations. It consists of studying the set of quantum correlations which are generated by actions on maximally entangled states. We call such correlations \textbf{maximally entangled correlations}. We show that the set of maximally entangled correlations, though not convex, is closed under rational convex combinations, and hence is dense in its own convex hull. Combined with previously known results, it follows that the maximally entangled synchronous correlations are dense in the set of all synchronous correlations. As a byproduct, we derive another formulation of Connes' embedding conjecture in terms of maximally entangled synchronous correlations.

Maximally entangled correlations have another interesting relationship with synchronous correlations. We define the \textbf{corner} of a synchronous correlation to be a subcorrelation formed by ``forgetting" the synchronized portion of the original correlation. It turns out that every local correlation is the corner of a synchronous local correlation and that every nonsignalling correlation is the corner of a synchronous nonsignalling correlation. This naturally leads to the question of whether or not every correlation set can be realized as the set of corners of the corresponding synchronous correlation set. We provide some partial answers in the cases of the quantum, quantum spacial, and quantum approximate correlations.

The idea of realizing a correlation as the corner of some larger symmetric correlation was explored by Sikora and Varvitsiotis  \cite{MR3612945} where different connections between correlation sets and positive semidefinite programming were explored. In particular, it was shown that there is a correspondence between quantum correlations and corners of certain doubly nonnegative matrices. We refer the reader to Section 3 of their paper  \cite{MR3612945} for more details.

Our paper is organized as follows. In section 2, we review definitions, notations and known results concerning correlation sets. In section 3, we introduce the set of maximally entangled correlations and discuss its main properties. Moreover, we provide a geometric description of the set of synchronous quantum correlations in terms of maximally entangled correlations and explore some consequences. In section 4, we establish the notion of corners and characterize the synchronous ones for various correlation sets. In section 5, we introduce a relaxed notion of maximally entangled correlation and discuss when these correlations can be approximated by maximally entangled correlations.

\section{Preliminaries}

In this section, we review the background material on quantum correlations. We follow conventions from the literature on synchronous correlations  \cite{MR3776034}. For a more thorough review of the foundations of quantum mechanics, we refer the reader to Nielsen and Chuang's textbook  \cite{MR1796805}. For more details concerning operator theory, we refer the reader to Davidson's textbook \cite{MR1402012}.

\subsection{Operator theory}

By a \textbf{Hilbert space}, we mean a complete complex vector space with a sesquilinear inner product. We let $B(H)$ denote the set of bounded linear operators on a Hilbert space $H$. When a Hilbert space $H$ is of finite dimension $d$, we may identify $H$ with the $d$-dimensional Euclidean space $\mathbb{C}^d$ and the set of linear operators on $H$ with the set of complex $d \times d$ matrices denoted by $\mathbb{M}_d$.

By a projection, we mean a linear operator $P$ on a Hilbert space $H$ satisfying $P^2=P$ and $P^*=P$, where $P^*$ denotes the adjoint operator of $P$. A finite set $\{P_k\}_{k=1}^m$ of positive operators on a Hilbert space $H$ is called a \textbf{positive operator-valued measure} if $\sum_{k=1}^m P_k = I_H$, where $I_H$ denotes the identity operator on $H$. If each $P_k$ is a projection, then $\{P_k\}_{k=1}^m$ is called a \textbf{projection-valued measure}.

Given two Hilbert spaces $H$ and $K$, we let $H \otimes K$ denote their Hilbert space tensor product. We will make use of the \textbf{Schmidt decomposition} (for example, Theorem A.5 of \cite{CLP2017}) of a vector $\phi \in H \otimes K$. In the case when $H$ and $K$ are finite-dimensional, any vector $\phi \in H \otimes K$ admits a Schmidt decomposition of the form $\phi = \sum_{k=1}^N \alpha_k e_k \otimes f_k$ for some $N$ and positive scalars $\{\alpha_k\}$. Here, $\{e_k\}$ and $\{f_k\}$ are orthonormal sets in $H$ and $K$ respectively.

Finally, we briefly introduce $C^*$-algebras and their states. For our purposes, a \textbf{$C^*$-algebra} is a unital closed subalgebra of $B(H)$ which is also closed under the adjoint operation. A \textbf{state on a $C^*$-algebra} $\mathfrak{A}$ is a linear map $\phi:\mathfrak{A} \rightarrow \complex$ which maps the unit of $\mathfrak{A}$ to 1 and maps positive operators to positive real numbers. A state $\phi$ on $\mathfrak{A}$ is \textbf{tracial} if $\phi(ab)=\phi(ba)$ for all $a,b \in \mathfrak{A}$. Since finite-dimensional $C^*$-algebras play an important role in this paper, we should mention that every finite-dimensional $C^*$-algebra is $*$-isomorphic to a finite direct sum of matrix algebras $\oplus_{k=1}^N \mathbb{M}_{d_k}$ (See Theorem III.1.1 of \cite{MR1402012}).

\subsection{Quantum mechanics}

The axioms of quantum mechanics dictate that a physical system corresponds to a Hilbert space $H$ and the state of a physical system corresponds to a unit vector in $H$. For this reason we use the terms state and unit vector interchangeably. A measurement on a physical system is given by a projection-valued measure $\{P_k\}_{k=1}^m$. The projections $P_1, P_2, \dots, P_m$ specify the possible outcomes of the measurement. When a physical system is in state $\phi$, the probability of observing outcome $k$ is given by $\langle P_k \phi, \phi \rangle$.

Given two physical systems $A$ and $B$ with corresponding Hilbert spaces $H_A$ and $H_B$, the state of the joint system is given by a unit vector $\phi \in H_A \otimes H_B$. When the two physical systems are non-interacting, the state of the joint physical system takes the form of a product state $\phi_A \otimes \phi_B$ for some unit vectors $\phi_A \in H_A$ and $\phi_B \in H_B$. Otherwise, the state takes the more general form described by the Schmidt decomposition and is considered to be entangled. In this case, local measurements on the separate physical systems are given by projection-valued measures of the form $\{P_k \otimes I_{H_B}\}$ and $\{I_{H_A} \otimes P_k\}$ respectively.

\subsection{Correlation sets}

Suppose two non-interacting players, Alice and Bob, each has a finite set of experiments with a finite number of outcomes. We let $p(i,j|x,y)$ represent the conditional probability that Alice performs experiment $x$ and gets outcome $i$ while Bob performs experiment $y$ and gets outcome $j$. The resulting tuple $p=\{p(i,j|x,y)\}$ is called a \textbf{correlation} if $p(i,j|x,y) \geq 0$ for all $i,j,x$ and $y$ and $\sum_{i,j} p(i,j|x,y) = 1$ for all $x$ and $y$. We let $n_A$ (resp. $n_B$) denote the number of Alice's (resp. Bob's) experiments and we let $m$ denote the number of possible outcomes per experiment. We let $C(n_A,n_B,m)$ denote the set of correlations for a given tuple $(n_A,n_B,m)$, and we write $C(n,m)$ for $C(n,n,m)$. Whenever the tuple $(n_A,n_B,m)$ is not specified or clear from context we simply write $C$ for a correlation set.

We will consider several particular correlation sets. The largest of these is the set of nonsignalling correlations, denoted by $C_{ns}$. A correlation $p$ is \textbf{nonsignalling} if the marginal densities defined by \[ p_A(i|x) := \sum_j p(i,j|x,y), \qquad p_B(j|y) := \sum_i p(i,j|x,y) \] are well defined, meaning that $\sum_j p(i,j|x,y)$ is independent of the choice of $y$ and $\sum_i p(i,j|x,y)$ is independent of the choice of $x$. The smallest of the correlation sets we will consider is the set of \textbf{local correlations} (or classical correlations), denoted by $C_{loc}$, which is defined to be the closed convex hull of the set of deterministic distributions $\{ p(i,j|x,y) : p(i,j|x,y) \in \{0,1\} \text{ for all } i,j,x,y\}$.


Between the local and nonsignalling correlation sets lies a variety of correlation sets whose definitions are inspired by problems in quantum mechanics. A correlation $p$ is called a \textbf{quantum correlation} if there exist finite-dimensional Hilbert spaces $H_A$ and $H_B$, projection-valued measures $\{E_{x,i}\}_{i=1}^m \subset B(H_A)$ and $\{F_{y,j}\}_{j=1}^m \subset B(H_B)$ for each $x \leq n_A$, $y \leq n_B$, and a unit vector $\phi \in H_A \otimes H_B$ such that \[ p(i,j|x,y) = \langle E_{x,i} \otimes F_{y,j} \phi, \phi \rangle. \] We refer to the tuple $(H_A,H_B,\{E_{x,i}\},\{F_{y,j}\},\phi)$ as a \textbf{representation} of $p$. If we relax the requirement that $H_A$ and $H_B$ be finite-dimensional, then we obtain the set of \textbf{quantum spacial correlations}. The closure of the set of quantum correlations is called the set of \textbf{quantum approximate correlations}. We denote by $C_q$ (resp. $C_{qs}$, $C_{qa}$) the set of quantum (resp. quantum spacial, quantum approximate). correlations.

Finally, we define the set of \textbf{quantum commuting correlations}, denoted by $C_{qc}$. A correlation $p$ is in $C_{qc}$ if there exists a Hilbert space $H$, projection-valued measures $\{E_{x,i}\}_{i=1}^m \subset B(H)$ and $\{F_{y,j}\}_{j=1}^m \subset B(H)$ satisfying $E_{x,i}F_{y,j} = F_{y,j}E_{x,i}$ for all $i,j,x$, and $y$, and a unit vector $\phi \in H$ such that \[ p(i,j|x,y) = \langle E_{x,i} F_{y,j} \phi, \phi \rangle. \]

A correlation is called \textbf{synchronous} if $n_A = n_B=n$ and if for each $x \leq n$ and $i \neq j$, we have $p(i,j|x,x)=0$. The set of synchronous correlations is distinguished from other correlation sets $C_r$ with the notation $C_r^s$.

\subsection{Known results}

It is well-known that the correlation sets satisfy the relations \[ C_{loc} \subseteq C_q \subseteq C_{qs} \subseteq C_{qa} \subseteq C_{qc} \subseteq C_{ns} \] and that they are all convex sets. It is also well-known that for certain choices of $n_A, n_B$, and $m$, we have $C_{loc} \neq C_q$ and $C_{qc} \neq C_{ns}$ (for example, see equation (2) and subsequent comments in \cite{DPP1}). Recently, Slofstra   \cite{Slofstra1} showed that $C_{qs} \neq C_{qa}$ in general, settling the so-called strong Tsirelson conjecture (see Remark 2.6 of \cite{MR3460238}). It is worth mentioning that at the time the current paper was being drafted, a preprint by Coladangelo and Stark   \cite{CqNeqCqs} found an example for the separation of the quantum and quantum spacial correlations, hence settling  $C_q \neq C_{qs}$. The only remaining inclusion is of particular importance, since the equality $C_{qa} = C_{qc}$ is known to be equivalent to the celebrated Connes' embedding conjecture (see \cite{MR3067294}, \cite{JMPPSW2011}, and \cite{Fritz2012}).

As in the non-synchronous case, the synchronous correlation sets satisfy \[ C_{loc}^s \subseteq C_q^s \subseteq C_{qs}^s \subseteq C_{qa}^s \subseteq C_{qc}^s \subseteq C_{ns}^s. \] Again, it is well-known that for certain choices of $n_A, n_B$ and $m$, we have $C_{loc}^s \neq C_q^s$ and $C_{qc}^s \neq C_{ns}^s$ (see for example equation (3) and subsequent comments in \cite{DPP1}). Dykema-Paulsen-Prakash showed  \cite{DPP1} that $C_{qs}^s \neq C_{qa}^s$ in general. In another recent paper  \cite{MR3776034}, Kim-Paulsen-Schafhauser showed that $C_q^s = C_{qs}^s$ and that $\overline{C_q^s} = C_{qa}^s$, settling a question posed by Dykema-Paulsen  \cite{MR3432742}. In the same paper, the authors also showed that $C_{qa}^s = C_{qc}^s$ is equivalent to Connes' embedding conjecture.

We will make extensive use of the following characterization of $C_q^s$.

\begin{theorem}[Paulsen, et. al.  \cite{MR3460238}, Theorem 5.5 / Corollary 5.6] \label{WinterThm} Let $p \in C_{q}(n,m)$. Then $p$ is a synchronous correlation if and only if there exist a finite-dimensional $C^*$-algebra $\mathfrak{A}$, projection-valued measures $\{E_{1,i}\},\dots,\{E_{n,i}\}$ in $\mathfrak{A}$ and a tracial state $\tau:\mathfrak{A} \rightarrow \mathbb{C}$ such that $p(i,j|x,y) = \tau(E_{x,i}E_{y,j})$. \end{theorem}

\section{Maximally entangled correlations}

Let $H_A$ and $H_B$ be finite-dimensional Hilbert spaces with $dim(H_A)=dim(H_B)=d$. Recall that a vector $\phi \in H_A \otimes H_B$ is called \textbf{maximally entangled} if \[ \phi = \frac{1}{\sqrt{d}} \sum_{k=1}^d u_k \otimes v_k \] for some orthogonal bases $\{u_k\}$ and $\{v_k\}$ of $H_A$ and $H_B$, respectively.

\begin{definition} \label{CmaxDefn} \emph{A quantum correlation $p$ is called \textbf{maximally entangled} if $p$ admits a representation $(H_A,H_B,\{E_{x,i}\},\{F_{y,j}\},\phi)$ where $\phi$ is a maximally entangled state. We denote by $C_{max}$ (resp. $C_{max}^s$) the set of maximally entangled correlations (resp. synchronous correlations). We write $C_{max,d}$ and $C_{max,d}^s$ when we want to emphasize that $dim(H_A)=dim(H_B)=d$.}\end{definition}

Already for dimension $d=1$ we begin to observe some interesting properties of the set $C_{max,d}$.

\begin{lemma} \label{lemmaCLoc} For all integers $n_A,n_B, n$ and $m$, we have $ext(C_{loc}(n_A,n_B,m)) = C_{max,1}$ and $ext(C_{loc}^s(n,m)) = C_{max,1}^s(n,m)$. \end{lemma}

\begin{proof} By definition, $C_{loc}$ is the convex hull of the deterministic correlations, where $p(i,j|x,y) \in \{0,1\}$. It is easily verified that every such correlation is of the form $p(i,j|x,y) = \delta_{x,i} \gamma_{y,i}$ where $\delta_{x,i}, \gamma_{y,j} \in \{0,1\}$ and $\sum_i \delta_{x,i} = \sum_j \gamma_{y,j} = 1$ for all $x$ and $y$. The synchronous case is similar. \end{proof}


The following lemma is well-known (for example, see the proof of proposition 3.4 in \cite{Slofstra1}). We will use it repeatedly throughout the paper - a proof is included for completeness.

\begin{lemma} \label{canonicalLemma} Let $p \in C_q$. Then $p \in C_{max,d}$ if and only if there exist projection-valued measures $\{E_{x,i}\}$ and $\{F_{y,j}\}$ in $\mathbb{M}_d$ for each $x$ and $y$ such that \[ p(i,j|x,y) = \frac{1}{d} Tr(E_{x,i}F_{y,j})\] for all $i,j,x,y$, where $Tr$ is the usual trace on $\mathbb{M}_d$.  Moreover, $p \in C_{max,d}^s$ if and only if the preceding statement holds for $F_{y,j}=E_{y,j}$ for every $y$ and $j$.\end{lemma}


\begin{proof} First assume $p \in C_{max,d}$. Let $(H_A,H_B,\{\tilde{E}_{x,i}\},\{\tilde{F}_{y,j}\},\tilde{\phi})$ be a representation of $p$ with $d$-dimensional Hilbert spaces $H_A$ and $H_B$ and $\tilde{\phi} = \frac{1}{\sqrt{d}} \sum_{k=1}^d u_k \otimes v_k$. Let $\{e_k\} \subset \mathbb{C}^d$ be the canonical orthonormal basis. Then there exist unitary matrices $U,V \in \mathbb{M}_d$ such that $Uu_k=Vv_k=e_k$ for each $k$. Define operators $E_{x,i} := U\tilde{E}_{x,i}U^*$ and $F_{y,j} := V \tilde{F}_{y,j}V^*$ for each $x,y,i$ and $j$, and a vector $\phi := (U \otimes V) \tilde{\phi}$. Clearly, $\{E_{x,i}\}$ and $\{F_{y,j}\}$ are projection-valued measures and $\phi$ is a maximally entangled state with a decomposition \[ \phi = \frac{1}{\sqrt{d}} \sum_{k=1}^d e_k \otimes e_k.\] Now, notice that \begin{eqnarray} p(i,j|x,y) & = & \langle \tilde{E}_{x,i} \otimes \tilde{F}_{y,j} \tilde{\phi}, \tilde{\phi} \rangle \nonumber \\ & = & \langle E_{x,i} \otimes F_{y,j} \phi, \phi \rangle \nonumber \\ & = & \frac{1}{d} Tr(E_{x,i} F_{y,j}^T) \nonumber \end{eqnarray} where $F_{y,j}^T$ is the transpose of $F_{y,j}$ with respect to the basis $\{e_k\}$. The reverse implication that $p(i,j|x,y) = \frac{1}{d} Tr(E_{x,i} F_{y,j})$ implies $p \in C_{max,d}$ can be easily verified by the reader using the final equality in the equation above.

In the synchronous case, observe that $p(i,j|x,x) = 0$ implies that $\frac{1}{d} Tr(E_{x,i} F_{x,j}) = 0$. However $Tr(AB) = 0$ implies $AB = 0$ for positive matrices $A$ and $B$. Hence $E_{x,i} F_{x,j} = 0$ for all $i \neq j$. Since $\sum_{i=1}^m E_{x,i} = \sum_{i=1}^m F_{x,i} = I_d$, \[ Tr(E_{x,i}^2) = Tr(E_{x,i} F_{x,i}) = Tr(F_{x,i}^2). \] Applying the Cauchy-Schwarz inequality to the inner product $\langle A, B \rangle = Tr(AB^*)$ we conclude that $E_{x,i} = F_{x,i}$. \end{proof}

In contrast with the correlation sets $C_r$ for $r \in \{loc,q,qs,qa,qc,ns\}$, the set $C_{max}$ is not convex. This is shown in the following proposition.

\begin{proposition} \label{rationalMarginal} Let $p \in C_{max}$. Then for each $i,j,x,$ and $y$, $p_A(i|x),p_B(j|y) \in \mathbb{Q}$.  Consequently the set $C_{max}$ is not convex. \end{proposition}

\begin{proof}  By Lemma \ref{canonicalLemma}, there exist projection-valued measures $\{E_{x,i}\},\{F_{y,j}\} \subset \mathbb{M}_d$ for some $d$ such that $p(i,j|x,y) = \frac{1}{d}Tr(E_{x,i}F_{y,j})$. Then for fixed $i$ and $x$, \[ p_A(i|x) = \sum_j p(i,j|x,y) = \frac{1}{d} Tr(E_{x,i}) \in \mathbb{Q} \] since $Tr(E_{x,i})$ is the rank of $E_{x,i}$. A similar calculation shows that $p_B(j|y) \in \mathbb{Q}$ for each fixed $j,y$.

 To see that $C_{max}$ is not convex, choose correlations $p^{(1)},p^{(2)} \in C_{max}$ such that $p_A^{(1)}(i|x) \neq p_A^{(2)}(i|x)$ and choose $r \in (0,1)$ such that $r p_A^{(1)}(i|x) + (1-r) p_A^{(2)}(i|x) \notin \mathbb{Q}$. Then $p^{(3)} := rp^{(1)} + (1-r)p^{(2)}$ is in $C_q$ since $C_q$ is convex, but $p^{(3)}_A(i|x)$ is irrational, proving that $p^{(3)} \notin C_{max}$.  \end{proof}



Despite not being convex in general, the next theorem shows that the set $C_{max}$ is closed under rational convex combinations.

\begin{theorem} \label{CmaxThm} For any collection $\{p_k\}_{k=1}^N$ of maximally entangled quantum correlations and any collection $\{t_k\}_{k=1}^N$ of positive rational numbers with $\sum_{k=1}^N t_k = 1$, we have \[ \sum_{k=1}^N t_k p_k \in C_{max}. \]  Moreover, if each $p_k$ is synchronous then $\sum_{k=1}^N t_k p_k \in C_{max}^s$. \end{theorem}

\begin{proof} Let $\{p_k\}_{k=1}^N \subset C_{max}$. Suppose that $p_k \in C_{max,d_k}$ for some integer $d_k$, for $k=1,\dots,N$. By Lemma \ref{canonicalLemma}, there exist projection-valued measures $\{E_{x,i}^{(k)}\}, \{F_{y,j}^{(k)}\} \subset \mathbb{M}_{d_k}$ such that each $p_k(i,j|x,y) = \frac{1}{d_k} Tr(E_{x,i}^{(k)}F_{y,j}^{(k)})$ for each $k$ and for all $i,j,x$ and $y$.

Since each $t_k$ is rational, there exist positive integers $M, n_1, \dots, n_N$ such that $t_k = \frac{n_k}{M}$. Define $R := \prod_{k=1}^N d_k$ and for each $k=1,\dots,N$ define $R_k := \frac{R}{d_k}$. Notice that \begin{eqnarray} \sum_{k=1}^N R_k d_k n_k & = & R (\sum_{k=1}^N n_k) \nonumber \\ & = & RM \nonumber \end{eqnarray} where the last line follows because $\sum_{k=1}^N t_k = \sum_{k=1}^N \frac{n_k}{M} = 1$. Consequently we obtain projection-valued measures $\{E_{x,i}\}, \{F_{y,j}\} \subset \mathbb{M}_{RM}$ by defining \[ E_{x,i} := \bigoplus_{k=1}^N \bigoplus_{l=1}^{R_k n_k} E_{x,i}^{(k)}, \qquad F_{y,j} := \bigoplus_{k=1}^N \bigoplus_{l=1}^{R_k n_k} F_{y,j}^{(k)} \] for each $i,j,x,$ and $y$.

Define a correlation $p:= \sum_{k=1}^N t_k p_k$.  For fixed $i,j,x$ and $y$, we have \[ p(i,j|x,y) = \frac{1}{RM} Tr(E_{x,i}F_{y,j}). \] By Lemma \ref{canonicalLemma}, $p \in C_{max,RM}$. The final statement regarding the synchronous case is clear. This completes the proof of the theorem.  \end{proof}



By Proposition \ref{rationalMarginal}, the set $C_{max}$ is not convex. We denote its convex hull by $co(C_{max})$. Similarly, we let $co(C_{max}^s)$ denote the convex hull of $C_{max}^s$. The next corollary is an easy consequence of Theorem \ref{CmaxThm}.

\begin{corollary} \label{coCmax} The set $C_{max}$ (resp. $C_{max}^s$) is dense in $co(C_{max})$ (resp. $co(C_{max}^s)$). \end{corollary}

\begin{proof} Let $p \in co(C_{max})$, and let $\epsilon > 0$ be given. Then $p = \sum_{k=1}^N t_k p_k$ with each $p_k \in C_{max}$ and $\sum_{k=1}^N t_k = 1$, $t_k \geq 0$. Choose positive rational numbers $r_k$ such that $\sum_{k=1}^N r_k = 1$, and $|t_k - r_k| < \frac{\epsilon}{N}$ for each $k$. Then $q := \sum_{k=1}^N r_k p_k \in C_{max}$ by Theorem \ref{CmaxThm}. Moreover, for any $i,j,x$ and $y$,\begin{eqnarray} |p(i,j|x,y) - q(i,j|x,y)| & = & |\sum_{k=1}^N (t_k-r_k) p_k(i,j|x,y)| \nonumber \\ & \leq & \sum_{k=1}^N |t_k - r_k| < \epsilon. \nonumber \end{eqnarray} Finally, notice that if $p \in co(C_{max}^s)$, then we could have chosen $p_k \in C_{max}^s$. In that case, $q \in C_{max}^s$ and the result follows. \end{proof}

In the synchronous scenario, there is a strong relationship between the sets $C_q^s$ and $C_{max}^s$. As the next theorem shows, it turns out that every synchronous quantum correlation can be approximated by a synchronous maximally entangled correlation.

\begin{theorem} \label{SyncQSEqualsCoCmax} The sets $C_{q}^s$ and $co(C_{max}^s)$ coincide. Consequently, $C_{max}^s$ is dense in $C_q^s$, and hence $\overline{C_{max}^s} = C_{qa}^s$. \end{theorem}

\begin{proof} Clearly, $co(C_{max}^s) \subseteq C_q^s$ since $C_{max}^s \subset C_q^s$ and $C_q^s$ is a convex set.

Conversely, let $p \in C_q^s$. By Theorem \ref{WinterThm}, there exist a finite-dimensional $C^*$-algebra $\mathfrak{A}$, a trace $\tau:\mathfrak{A} \rightarrow \mathbb{C}$ and projection-valued measures $\{E_{x,i}\} \subset \mathfrak{A}$ such that $p(i,j|x,y)=\tau(E_{x,i}E_{y,j})$ for all $i,j,x$ and $y$. Since $\mathfrak{A}$ is finite-dimensional, we may assume that $\mathfrak{A} = \bigoplus_{k=1}^M \mathbb{M}_{n_k}$ for some positive integers $n_1,\dots, n_M$. For each $k=1,\dots,M$, let $E_{x,i}^{(k)}$ be the projection of $E_{x,i}$ onto the $k$-th summand. Then the $\{E_{x,i}^{(k)}\}$'s are projection-valued measures on $\mathbb{M}_{n_k}$.

Defining $p_k(i,j|x,y) := \frac{1}{n_k}Tr(E_{x,i}^{(k)}E_{y,j}^{(k)})$, it is easy to see by Lemma \ref{canonicalLemma} that $p_k \in C_{max}^s$ for each $k$. Moreover, by Example IV.5.4 of Davidson's textbook  \cite{MR1402012}, there exist positive scalars $t_1, \dots, t_M$ such that $\sum_{k=1}^M t_k = 1$ and \begin{eqnarray} p(i,j|x,y) & = & \tau(E_{x,i}E_{y,j}) \nonumber \\ & = & \sum_{k=1}^M \frac{t_k}{n_k} Tr(E_{x,i}^{(k)}E_{y,j}^{(k)}) \nonumber \\ & = & \sum_{k=1}^M t_k p_k(i,j|x,y). \nonumber \end{eqnarray} Thus, $p \in co(C_{max}^s)$.

The final statement follows from Corollary \ref{coCmax} and Theorem III.6 of Kim-Paulsen-Schafhauser's paper  \cite{MR3776034} that $\overline{C_q^s} = C_{qa}^s$. \end{proof}

After writing the initial draft of this paper, the authors learned that the relation $C_q^s = co(C_{max}^s)$ appears in some other preprints, namely Theorem 9 of   \cite{Lackey} and Corollary 5.5 of   \cite{LupiniEtAlPerfect2020}. The density result appears to be new.

Since $C_{max}^s$ is dense in $C_q^s$, it is natural to ask if the same relationship holds for non-synchronous correlations. We will explore this question further in the next section.

As a byproduct of Theorem \ref{SyncQSEqualsCoCmax}, we obtain a new reformulation of Connes' embedding conjecture in terms of the set of maximally entangled quantum correlations. By Corollary III.8 of Kim-Paulsen-Schafhauser's paper  \cite{MR3776034}, Connes' embedding conjecture is equivalent to the assertion that $C_{qc}^s(n,m) = C_{qa}^s(n,m)$ for all $n$ and $m$. By the previous theorem, we know that $\overline{C_{max}^s} = C_{qa}^s$. Combining these two results, we arrive at the following corollary.

\begin{corollary} \label{ConnesThm} The following statements are equivalent. \begin{enumerate} \item Connes' embedding conjecture is true. \item $C_{qc}^s(n,m) = \overline{C_{max}^s(n,m)}$ for all $n$ and $m$. \end{enumerate} \end{corollary}

We conclude this section by constructing a nested sequence of closed convex sets which form an inner approximation of $C_q^s$.

\begin{theorem} For each positive integer $k$, define $C_k^s := co(C_{max,k!}^s)$. Then for each $k$, $C_k^s$ is closed and convex, $C_k^s \subseteq C_{k+1}^s$, and $\cup_{k=1}^\infty C_k^s = C_q^s$. \end{theorem}

\begin{proof} It is obvious that each $C_k^s$ is convex. To see that each $C_k^s$ is closed, it suffices to show that $C_{max,d}^s$ is closed for each $d$. To this end, suppose that $\{p_k\}_{k=1}^\infty \subseteq C_{max,d}^s$ is a convergent sequence with $p_k \rightarrow p \in C_{qa}^s$. We will show that $p \in C_{max,d}^s$. To see this, first note that for any $q \in C_{ns}^s$, $q(i,i|x,x) = \sum_j q(i,j|x,x) = q_A(i|x)$. Hence, Proposition \ref{rationalMarginal} implies that for each $k,i,$ and $x$, $p_k(i,i|x,x) \in \{0, 1/d, 2/d, \dots, 1\}$. Since the sequence $\{p_k(i,i|x,x)\}$ converges, it must be constant after finitely many terms. Thus we may assume that $p_k(i,i|x,x) = n_{i,x}/d$ for every choice of $i$ and $x$.

Since $p_k \in C_{max,d}^s$, Lemma \ref{canonicalLemma} says we may choose for each $k$ projection-valued measures $\{E_{x,i}^{(k)}\} \subset \mathbb{M}_d$ such that $p_k(i,j|x,y) = \frac{1}{d} Tr(E_{x,i}^{(k)} E_{y,j}^{(k)})$, and we may further assume that $Rank(E_{x,i}^{(k)}) = n_{x,i}$ for all $k$. Since the set of $d \times d$ projection matrices of rank $n_{x,i}$ is compact, there exists a rank $n_{x,i}$ projection $E_{x,i}$ which is a limit point of the set of projections $\{E_{x,i}^{(k)}\}$. By the continuity of the trace, we see that $p(i,j|x,y) = \frac{1}{d} Tr(E_{x,i}E_{y,j})$. Moreover, since $E_{x,i}^{(k)}$ converges to $E_{x,i}$, we have $I_d = \sum_{i=1}^m E_{x,i}^{(k)}$ converges to $\sum_{i=1}^m E_{x,i}$ and hence $\{E_{x,i}\}$ is a projection valued measure.

To see that $C_k^s \subseteq C_{k+1}^s$, it suffices to show that $C_{max,k!}^s \subset C_{max,(k+1)!}^s$. Pick $p \in C_{max,k!}^s$ with $p(i,j|x,y) = \frac{1}{k!} Tr(E_{x,i} E_{y,j})$. Then \[ p(i,j|x,y) = \sum_{n=1}^{k+1} \frac{1}{k+1} p(i,j|x,y) = \frac{1}{(k+1)!} Tr(\oplus_{n=1}^{k+1} E_{x,i} \oplus_{n=1}^{k+1} E_{y,j}). \] Hence, $p \in C_{max,(k+1)!}^s$.

Finally, observe that by Theorem \ref{SyncQSEqualsCoCmax}, $C_q^s = co(C_{max}^s)$. We need only show that $co(C_{max}^s) = \cup_{k=1}^\infty C_k^s$. Obviously $\cup_{k=1}^\infty C_k^s \subseteq co(C_{max}^s)$. To see the other inclusion, pick $p \in C_{max}^s$. Then $p \in C_{max,d}^s$ for some $d$. By Lemma \ref{canonicalLemma}, there exists $\{E_{x,i}\} \subseteq \mathbb{M}_d$ such that $p(i,j|x,y) = \frac{1}{d} Tr(E_{x,i} E_{y,j})$. Hence, $p \in C_{max,d!}^s$, since \[ p(i,j|x,y) = \sum_{n=1}^{(d-1)!} \frac{1}{(d-1)!} p(i,j|x,y) \nonumber = \frac{1}{d!} Tr(\oplus_{n=1}^{(d-1)!} E_{x,i} \oplus_{n=1}^{(d-1)!} E_{y,j}). \] So $p \in C_d^s$. Thus, $C_{max}^s \subset \cup_k C_k^s$. By convexity, $co(C_{max}) \subseteq \cup_{k=1}^\infty C_k^s$, completing the proof. \end{proof}

We remark that $C_1^s = C_{loc}^s$, since $C_1^s$ is defined to be $ext(C_{max,1}^s)$ and $ext(C_{max,1}^s)=C_{loc}^s$ by Lemma \ref{lemmaCLoc}. Thus, the sequence $\{C_k^s\}$ satisfies \[ C_{loc}^s = C_1^s \subseteq C_2^s \subseteq \dots \subset C_q^s, \qquad \bigcup_{k=1}^\infty C_k^s = C_q^s. \] We hope that this observation may someday shed light on the question of why $C_q^s$ is not closed in general.

\section{Corners of synchronous correlations}

Given a synchronous correlation $p$, we wish to define a sub-correlation $\pi(p)$ of $p$ which captures the interaction of some non-synchronous partition of Alice and Bob's experiments. We call this sub-correlation the corner of $p$. To this end, we make the following definition.

\begin{definition} \label{cornerDefn} Let $n_A,n_B$, and $m$ be positive integers and set $n := n_A + n_B$. Then for each $r \in \{loc, max, q, qs, qa, qc, ns\}$, define the projection map $\pi_{n_A,n_B}:C_r(n,m) \rightarrow C_r(n_A,n_B,m)$ by $\pi_{n_A,n_B}(p)(i,j|x,y) = p(i,j|x,y+n_A)$ for each $1 \leq i,j \leq m$, $1 \leq x \leq n_A$ and $1 \leq y \leq n_B$. We call $\pi_{n_A,n_B}(p)$ the \textbf{corner} of $p$. When the indices $n_A$ and $n_B$ are understood from context, we simply write $\pi(p)$. \end{definition}

When $p \in C_r(n,m)$, we may regard $p$ as a block matrix \[ \begin{bmatrix} A & B \\ C & D \end{bmatrix} \] where the dimensions of $A, B, C,$ and $D$ are $mn_A \times mn_A$, $mn_A \times mn_B$, $mn_B \times mn_A$, and $mn_B \times mn_B$, respectively. With this notation, we can identify $\pi(p)$ with the matrix $B$.

We are interested in studying the corners of synchronous correlations. It is easy to see that, for each $r \in \{loc, max, q, qs, qa, qc, ns\}$, $\pi(C_r^s(n,m)) \subseteq C_r(n_A,n_B,m)$ for all $n_A,n_B,n$ and $m$ satisfying the conditions of Definition \ref{cornerDefn}. We raise the following question.

\begin{question} \label{mainQuestion} When is $\pi(C_r^s(n,m)) = C_r(n_A,n_B,m)$? \end{question}

Our interest is partly motivated by the following observation concerning $C_{qs}$. Recently, a surprising result of Kim-Paulsen-Schafhauser   \cite{MR3776034} (Theorem III.10) showed that the sets $C_q^s$ and $C_{qs}^s$ coincide. However, it seems that this does not hold in the non-synchronous case   \cite{CqNeqCqs}. If it were the case that $\pi(C_{qs}^s) = C_{qs}$, we see that \[ C_{qs} = \pi(C_{qs}^s) = \pi(C_q^s) \subseteq C_q,\] and hence $C_q = C_{qs}$.

In the remainder of this section, we will answer Question \ref{mainQuestion} in the affirmative for the cases $r \in \{loc, max, ns\}$ and provide partial results in the other cases. We begin with a crucial proposition.

\begin{proposition} \label{CornerLemmaCmax} For every positive integer $d$, $\pi(C_{max,d}^s) = C_{max,d}$. Consequently $\pi(co(C_{max,d}^s)) = co(C_{max,d})$, $\pi(C_{max}^s)=C_{max}$, $\pi(co(C_{max}^s)) = co(C_{max})$, and $\pi(\overline{C_{max}^s}) = \overline{C_{max}}$. \end{proposition}

\begin{proof} We show that $C_{max,d} \subseteq \pi(C_{max,d}^s)$. To this end, suppose that $p \in C_{max,d}$. By Lemma \ref{canonicalLemma} there exist projection-valued measures $\{E_{x,i}\}$ and $\{F_{y,j}\}$ in $\mathbb{M}_d$ such that $p(i,j|x,y) = \frac{1}{d} Tr(E_{x,i}F_{y,j})$. For each $x \leq n_A$, set $\tilde{E}_{x,i} = \tilde{F}_{x,i} = E_{x,i}$, and for each $x$ satisfying $n_A < x \leq n_A + n_B$ set  $\tilde{E}_{x,i} = \tilde{F}_{x,i} := F_{x-n_A,i}$. Setting $\tilde{p}(i,j|x,y) := \frac{1}{d}Tr(\tilde{E}_{x,i}\tilde{F}_{y,j})$ defines a correlation in $C_{max,d}$ satisfying $\pi(\tilde{p}) = p$. Moreover, $\tilde{p}$ is synchronous since for each $x \leq n$ and $i \neq j$, we have \[ \tilde{p}(i,j|x,x) = \frac{1}{d}(\tilde{E}_{x,i}\tilde{F}_{x,j}) = 0 \] since $\tilde{E}_{x,i}\tilde{F}_{x,j}$ equals either $E_{x,i}E_{x,j}$ or $F_{x,i}F_{x,j}$, both of which are zero. So $C_{max,d} \subseteq \pi(C_{max,d}^s)$ and consequently $\pi(C_{max,d}^s)=C_{max,d}$.  The remaining statements follow from the observation that $C_{max} = \cup_{d=1}^\infty C_{max,d}$ and the fact that $\pi$ is continuous and affine.   \end{proof}


The above proposition answers Question \ref{mainQuestion} in the affirmative for the case $r=max$. While a complete answer to Question \ref{mainQuestion} is unknown, the next theorem summarizes some partial results towards this end.

\begin{theorem} For all positive integers $n_A, n_B, n$, and $m$ satisfying $n_A + n_B = n$, the following statements are true. \begin{enumerate} \item The sets $\pi(C_{loc}^s(n,m))$ and $C_{loc}(n_A,n_B,m)$ coincide. \item The sets $\pi(C_q^s(n,m))$ and $C_q(n_A,n_B,m)$ coincide if and only if \[C_q(n_A,n_B,m)=co(C_{max}(n_A,n_B,m)).\] \item The sets $\pi(C_{qa}^s(n,m))$ and $C_{qa}(n_A,n_B,m)$ coincide if and only if \[ C_{qa}(n_A,n_B,m) = \overline{C_{max}(n_A,n_B,m)}. \] \end{enumerate} \end{theorem}

\begin{proof} (1) Recall that $C_{loc} = co(C_{max,1})$ and $C_{loc}^s = co(C_{max,1}^s)$, by Lemma \ref{lemmaCLoc}. Consequently $\pi(C_{loc}^s) = \pi(co(C_{max,1}^s)) = co(C_{max,1})=C_{loc}$ by Proposition \ref{CornerLemmaCmax}.


(2) By Theorem \ref{SyncQSEqualsCoCmax}, we know that $C_q^s = co(C_{max}^s)$. Therefore $C_q = \pi(C_{q}^s)$ if and only if $C_q=co(C_{max})$ since $\pi(C_q^s) = \pi(co(C_{max}^s)) = co(C_{max})$ by Proposition \ref{CornerLemmaCmax}.


(3) By Theorem \ref{SyncQSEqualsCoCmax}, we know that $C_{qa}^s = \overline{C_{max}^s}$. By Proposition \ref{CornerLemmaCmax}, $\pi(C_{qa}^s) = \pi(\overline{C_{max}^s}) = \overline{C_{max}}$. The statement follows. \end{proof}


We conclude with a proof that $\pi(C_{ns}^s) = C_{ns}$. In the following, we call a correlation $p \in C_r(n,m)$ \textbf{symmetric} if $p(i,j|x,y) = p(j,i|y,x)$ for all $i,j,x$ and $y$. The next theorem shows that to each $p \in C_{ns}$ there exists a symmetric $\tilde{p} \in C_{ns}^s$ such that $p = \pi(\tilde{p})$. The ability to choose a symmetric correlation $\tilde{p}$ is worth noting because every correlation in $C_r^s$ for $r \in \{loc, q, qs, qa, qc\}$ is known to be symmetric   \cite{MR3356844}, though this need not hold for correlations in $C_{ns}^s$.

\begin{theorem} The sets $\pi(C_{ns}^s)$ and $C_{ns}$ coincide. In particular, every correlation in $C_{ns}$ is the corner of some symmetric correlation in $C_{ns}^s$.\end{theorem}

\begin{proof} Let $p \in C_{ns}(n_A,n_B,m)$. We proceed by constructing a symmetric correlation $\tilde{p} \in C_{ns}^s(n,m)$ which satisfies $\pi(\tilde{p})=p$, where $n = n_A + n_B$ as usual.

As before, let $p_A(i|x) = \sum_j p(i,j|x,y)$ and $p_B(j|y) = \sum_i p(i,j|x,y)$ denote the marginal densities of $p$. Recall that the nonsignalling conditions dictate that these sums are well-defined. We begin by defining correlations $p_1 \in C_{ns}^s(n_A, m)$ and $p_2 \in C_{ns}^s(n_B,m)$ as follows. For each $x \leq n_A$ and $i,j \leq m$, set
\[ p_1(i,j|x,x) = \begin{cases}
      p_A(i|x) & i = j \\
      0 & i \neq j \\
   \end{cases} \] and for each $y \leq n_A$ with $x \neq y$, define $p_1(i,j|x,y) := p_A(i|x)p_A(j|y)$. Likewise, we define $p_2$ so that for each $y \leq n_B$ and $i,j \leq m$, \[ p_2(i,j|y,y) := \begin{cases}
      p_B(i|y) & i = j \\
      0 & i \neq j \\
   \end{cases} \] and for each $x \leq n_B$ with $x \neq y$, define $p_2(i,j|x,y) := p_B(i|x)p_B(j|y)$. It is clear that $p_1$ and $p_2$ are symmetric and synchronous.  The reader can easily verify that they satisfy the nonsignalling conditions as well. 
  
    Finally we define a correlation $\tilde{p}$ as follows. For each $x,y \leq n$ and $i,j \leq m$, set \[ \tilde{p}(i,j|x,y) := \begin{cases} p_1(i,j|x,y) & x,y \leq n_A \\  p(i,j|x,y - n_A) & x \leq n_A, n_A < y \leq n \\  p(j,i|y,x - n_A) & n_A < x \leq n, y \leq n_A \\ p_2(i,j|x-n_A,y-n_A) & n_A < x,y \leq n   \end{cases} \qquad . \] It is easy to verify that the correlation $\tilde{p}$ is symmetric, synchronous and nonsignalling. Moreover, $\pi(\tilde{p}) = p$.\end{proof}

\section{Almost maximally entangled correlations}

In this section we consider correlations obtained from a slight relaxation of Definition \ref{CmaxDefn}.

\begin{definition} \emph{A quantum correlation $p$ is called \textbf{almost maximally entangled} if $p$ admits a representation $(H_A, H_B, \{P_{x,i}\}, \{Q_{y,j}\}, \phi)$ where $\phi$ is maximally entangled and $\{P_{x,i}\}, \{Q_{y,j}\}$ are positive operator-valued measures (not necessarily projections). We denote by $C'_{max}$ the set of almost maximally entangled correlations.} \end{definition}

The next lemma generalizes Lemma \ref{canonicalLemma} to the setting of almost maximally entangled correlations. 

\begin{lemma} \label{POVMLemma} Every almost maximally entangled correlation is a quantum correlation. In particular, following statements are equivalent.
\begin{enumerate}
    \item $p \in C'_{max}(n_A,n_B,m)$.
    \item There exist positive operator-valued measures $\{P_{x,i}\}_{i=1}^m$ and $\{Q_{y,j}\}_{j=1}^m$ on $\mathbb{C}^d$ for $x \leq n_A, y \leq n_B$ such that $p(i,j|x,y) = \frac{1}{d} Tr(P_{x,i} Q_{y,j})$.
    \item The correlation $p$ has a representation $(H_A, H_B, \{E_{x,i}\}, \{F_{y,j}\}, \phi)$ where $\phi$ is a direct sum of a maximally entangled vector with the zero vector and the operators $E_{x,i}$ and $F_{y,j}$ are projections for all $x,y,i,j$.
\end{enumerate}
\end{lemma}

\begin{proof} The equivalence of the first two statements can be showing by repeating the calculation in the proof of Lemma \ref{canonicalLemma}. We show that the first and the last statements are equivalent.

First, assume that the third statement holds. We may assume that \[\phi = \sum_{i=1}^k \frac{1}{\sqrt{k}} e_i \otimes f_i \] for some $k \leq d$ and orthonormal sets of vectors $\{e_1,e_2,\dots,e_k\} \subset H_A$ and $\{f_1, f_2, \dots, f_k \} \subset H_B$. Let $P$ be the projection onto the span of $\{e_1,e_2,\dots,e_k\}$ and let $Q$ be the projection onto the span of $\{f_1, f_2, \dots, f_k \}$. Then \[ \langle E_{x,i} \otimes F_{y,j} \phi, \phi \rangle = \langle (PE_{x,i}P) \otimes (Q F_{y,j} Q) \phi, \phi \rangle. \] The first statement follows by replacing $H_A$ and $H_B$ with the ranges of $P$ and $Q$, respectively.

Next assume that the first statement holds. By the arguments in the proof of Theorem 5.3 of   \cite{MR3460238}, there exist Hilbert spaces $K_A$ and $K_B$ such that $H_A \subseteq K_A$ and $H_B \subseteq K_B$, and projection-valued measures $\{E_{x,i}\}$ on $K_A$ and $\{F_{y,j}\}$ on $K_B$ such that \[ \langle P_{x,i} \otimes Q_{y,j} \phi, \phi \rangle = \langle E_{x,i} \otimes F_{y,j} \phi, \phi \rangle \] where we have identified $\phi$ with $\phi \oplus 0$ in $K_A \otimes K_B$. Since $\phi$ was maximally entangled in $H_A \otimes H_B$, the third statement follows. \end{proof}

Though we do not know whether or not $C_{max}$ is dense in $C'_{max}$ in general, the remaining theorems provide partial answers towards this question. The next proposition is the crux of our argument.

\begin{proposition} \label{CmaxToCmaxPrime} Suppose that $p \in C'_{max}(n_A,n_B,m)$ satisfies $p(i,j|x,y) = \frac{1}{d} Tr(P_{x,i} Q_{y,j})$ for positive operator-valued measures $\{P_{x,i}\}$ and $\{Q_{y,j}\}$ on $\mathbb{C}^d$. Also, suppose that $P_{x,i}P_{x,j} = P_{x,j}P_{x,i}$ and $Q_{y,i}Q_{y,j} = Q_{y,j}Q_{y,i}$ for each $x,y,i,j$. If each operator $P_{x,i}$ and $Q_{y,j}$ has rational eigenvalues, then $p \in C_{max}(n_A,n_B,m)$. \end{proposition}

\begin{proof} Choose some $z \leq n_A$. Since $P_{z,i}P_{z,j} = P_{z,j}P_{z,i}$, the matrices $P_{z,i}$ are simultaneously diagonalizable, so that $P_{z,i} = U^* P'_{z,i} U$ for diagonal $P'_{z,i}$ and a unitary $U$. By conjugating each $P_{x,i}$ and $Q_{y,j}$ by $U$, we obtain new positive operator-valued measures $\{P'_{x,i}\}$ and $\{Q'_{y,j}\}$ which still satisfy the conditions of the proposition and yield the same correlation - in particular, their eigenvalues are unchanged. Choose a common denominator $N$ for the diagonal entries of $P'_{z,1},\dots,P'_{z,m}$, so that each is of the form $\frac{k}{N}$. For each $x \neq z$ and each $y$, define $\tilde{P}_{x,i} := I_{N} \otimes P'_{x,i}$ and $\tilde{Q}_{y,j} := I_N \otimes Q'_{y,j}$. Then the positive operator-valued measures $\{\tilde{P}_{x,i}\}$ and $\{\tilde{Q}_{y,j}\}$ still satisfy the conditions of the proposition with unchanged eigenvalues (ignoring multiplicity). 

We now construct a projection-valued measure $\{E_{z,i}\}$ as follows. Assume the diagonal entries of $P'_{z,i}$ are $\frac{n_{1,i}}{N}, \frac{n_{2,i}}{N}, \dots, \frac{n_{d,i}}{N}$. Since $\sum_i P'_{z,i} = I_d$, we see that $\sum_i n_{k,i} = N$ for each $k \leq d$. Hence, there exist projection-valued measures $\{G_{k,i}\}_{i=1}^m$ on $\mathbb{C}^N$ such that the rank of $G_{k,i}$ is $n_{k,i}$. For each $i \leq m$, set $E_{z,i} = \oplus_{k=1}^d G_{k,i}$. Then $\{E_{z,i}\}$ is a projection-valued measure on $\mathbb{C}^N \otimes \mathbb{C}^d$. Now observe that \begin{eqnarray} \frac{1}{d} Tr(P'_{z,i} Q'_{y,j}) & = & \frac{1}{d} \sum_{k=1}^d \frac{n_{k,i}}{N} Q'_{y,j}(k,k) \nonumber \\ & = & \frac{1}{dN} \sum_{k=1}^d Tr(G_{k,i})Q'_{y,j}(k,k) \nonumber \\ & = & \frac{1}{dN} Tr(E_{z,i} \tilde{Q}_{y,j}) \nonumber \end{eqnarray} where $Q'_{y,j}(k,k)$ is the $k$-th diagonal entry of $Q'_{y,j}$. It is obvious that $\frac{1}{d} Tr(P_{x,i} Q_{y,j}) = \frac{1}{dN} Tr(\tilde{P}_{x,i} \tilde{Q}_{y,j})$ for each $x \neq z$. Hence we may have assumed without loss of generality that $\{P_{z,i}\}_i$ was a projection-valued measure. Repeating this argument for each $\{P_{x,i}\}$, as well as each $\{Q_{y,j}\}$, proves that $p \in C_{max}(n_A,n_B,m)$. \end{proof}

\begin{theorem} \label{CmaxDenseCPrime} Let $p \in C'_{max}(n_A,n_B,m)$ satisfy $p(i,j|x,y) = \frac{1}{d} Tr(P_{x,i} Q_{y,j})$ for positive operator-valued measures $\{P_{x,i}\}$ and $\{Q_{y,j}\}$ on $\mathbb{C}^d$. Suppose that $P_{x,i}P_{x,j} = P_{x,j}P_{x,i}$ and $Q_{y,i}Q_{y,j} = Q_{y,j}Q_{y,i}$ for each $x,y,i,j$. Then for every $\epsilon > 0$, there exists a $q \in C_{max}(n_A,n_B,m)$ such that $|p(i,j|x,y) - q(i,j|x,y)| < \epsilon$. \end{theorem}

\begin{proof} Pick some $z \leq n_A$. Since $P_{z,i}$ commutes with $P_{z,j}$ for each $i,j \leq m$, we can simultaneously diagonalize each $P_{z,i}$, so that $P_{z,i} = U^* P'_{z,i} U$ with $P'_{z,i}$ diagonal. Assume the $k$-th diagonal entry of $P'_{z,i}$ is $\lambda_{k,i}$. Choose positive rational numbers $r_{k,i}$ such that $\sum_i r_{k,i} = 1$ and $|\lambda_{k,i} - r_{k,i}| < \frac{\epsilon}{2}$. Set $\tilde{P}_{z,i} = U^*D_{z,i}U$ where $D$ is diagonal with entries $r_{1,i}, r_{2,i}, \dots, r_{d,i}$. Then for each $y$ and $j$ \begin{eqnarray} |\frac{1}{d} Tr(P_{z,i} Q_{y,j}) - \frac{1}{d} Tr(\tilde{P}_{z,i} Q_{y,j}) | & = & \frac{1}{d} | Tr((P_{z,i} - \tilde{P}_{z,i})Q_{y,j}) | \nonumber \\ &  \leq  &  \frac{1}{d} \sum_{k=1}^d |\lambda_{k,i} - r_{k,i}|  \nonumber \\ & \leq & \frac{1}{d} \sum_{k=1}^d \frac{\epsilon}{2} = \frac{\epsilon}{2}. \nonumber \end{eqnarray} Replacing all the $P_{x,i}$'s with $\tilde{P}_{x,i}$'s in this way gives a new correlation $q'$ in $C'_{max}$ with $|p(i,j|x,y) - q'(i,j|x,y)| < \epsilon/2$, and the $\tilde{P}_{x,i}$ have rational eigenvalues. Similarly, we can replace the $Q_{y,j}$'s with $\tilde{Q}_{y,j}$'s such that each $\tilde{Q}_{y,j}$ has rational eigenvalues and the correlation $q$ given by $q(i,j|x,y) = \frac{1}{d} Tr(\tilde{P}_{x,i} \tilde{Q}_{y,j})$ satisfies $|q'(i,j|x,y) - q(i,j|x,y)| < \epsilon/2$. It follows that $|p(i,j|x,y) - q(i,j|x,y)| < \epsilon$. Finally, we see that $q \in C_{max}(n_A,n_B,m)$ by Proposition \ref{CmaxToCmaxPrime}. \end{proof}

\begin{corollary} The set $C_{max}(n_A,n_B,2)$ is dense in $C'_{max}(n_A,n_B,2)$. \end{corollary}

\begin{proof} In the case $m=2$, a positive operator-valued measure consists of a pair $P_{x,1},P_{x,2}$ of positive operators satisfying $P_{x,1} + P_{x,2} = I_d$. Since $P_{x,2} = I_d - P_{x,1}$, $P_{x,1}$ commutes with $P_{x,2}$. The statement then follows from Theorem \ref{CmaxDenseCPrime}. \end{proof}


\section*{Acknowledgements}

We would like to thank Christopher Schafhauser for his helpful advice throughout the course of this research. We are also grateful to Vern Paulsen and William Slofstra for interesting conversations and helpful remarks. In particular, we should acknowledge that the questions considered in section 5 were partly inspired by a lunch conversation between the authors and William Slofstra at the Institute for Pure and Applied Mathematics workshop ``Approximation Properties in Operator Algebras and Ergodic Theory" in May of 2018. Finally, we would like to express our gratitude to the referee for many helpful comments which led to significant improvements in the exposition.


\bibliographystyle{alpha}
\bibliography{references}

\end{document}